\documentclass[a4paper,USenglish,cleveref, autoref, thm-restate]{amsart}

\title{Local Computation Algorithms for Coloring of Uniform Hypergraphs}

\author{Andrzej Dorobisz}
\address{Theoretical Computer Science Department, Faculty of Mathematics and Computer Science, Jagiellonian University, Krak\'{o}w, Poland}
\email{andrzej.dorobisz@tcs.uj.edu.pl}

\author{Jakub Kozik}
\address{Theoretical Computer Science Department, Faculty of Mathematics and Computer Science, Jagiellonian University, Krak\'{o}w, Poland}
\email{jakub.kozik@uj.edu.pl}


\thanks{This work was partially supported by Polish National Science Center (2016/21/B/ST6/02165)}

\usepackage[ruled,vlined,titlenumbered,linesnumbered,noend,norelsize]{algorithm2e}
\DontPrintSemicolon
\SetNlSty{footnotesize}{}{}
\IncMargin{0.5em}
\SetNlSkip{1em}
\SetKwIF{If}{ElseIf}{Else}{if}{then}{else if}{else}{endif}
\SetKw{Return}{return}
\SetKw{Pick}{pick}
\SetKw{FAIL}{FAIL}
\SetNlSkip{2em}
\newenvironment{algorithm-hbox}{\hbadness=10000\begin{algorithm}}{\end{algorithm}}

\SetAlgorithmName{Listing}{Listings}

\newcommand{\E}[1]{\mathbb{E}{[#1]}}
\newcommand{\Et}[1]{\E{\textrm{#1}}}
\newcommand{\pr}[1]{\Pr{[#1]}}
\newcommand{\prt}[1]{\pr{\textrm{#1}}}

\newcommand{\eps}{\varepsilon}
\newcommand{\eul}{\mathrm{e}}
\newcommand{\emptyseq}{\epsilon}

\newcommand{\ceil}[1]{\lceil{#1}\rceil}

\newcommand{\Oh}[1]{\mathcal{O}{\left(#1\right)}}
\newcommand{\oh}[1]{o{\left(#1\right)}}

\newcommand{\event}[1]{\mathcal{X}(#1)}

\newcommand{\polylog}{\mathrm{polylog}}

\newcommand{\fun}[1]{\textsc{#1}{}}                     
\newcommand{\state}[1]{\texttt{#1}}                     

\newcommand{\pBadTail}{P_{W}}
\newcommand{\pMonoTail}{P_{M}}

\newcommand{\rc}{\mathcal{C}}           
\newcommand{\Hrc}{H_{\alpha}(\rc)}      

\newcommand{\rcol}{\rho}

\newcommand{\type}[1]{(#1)}

\newcommand{\ort}{\phi}
\newcommand{\lab}{\sigma}

\newcommand{\alStar}{\alpha_\star}

\newcommand{\vCompBound}{2k\log(m)}  

\theoremstyle{plain}
\newtheorem{theorem}{Theorem}
\newtheorem{observation}[theorem]{Observation}
\newtheorem{remark}[theorem]{Remark}
\newtheorem{definition}[theorem]{Definition}
\newtheorem{proposition}[theorem]{Proposition}
\newtheorem{corollary}[theorem]{Corollary}
\newtheorem{claim}[theorem]{Claim}

\usepackage[margin=3cm]{geometry}
\usepackage{enumerate}

\begin{document}

\maketitle

\begin{abstract}
We present a progress on local computation algorithms for two coloring of $k$-uniform hypergraphs.
    We focus on instances that satisfy strengthened assumption of Local Lemma of the form
    $
        2^{1-\alpha k} (\Delta+1) \eul < 1,
    $
    where $\Delta$ is the bound on the maximum edge degree of the hypergraph.
    We discuss how previous works on the subject can be used to obtain an algorithm that works in polylogarithmic time per query for
        $\alpha$ up to about $0.139$.
Then, we present a procedure that, within similar bounds on running time, solves wider range of instances by allowing $\alpha$ at most about $0.227$.
\end{abstract}

\section{Introduction}
For a long time the problem of hypergraph coloring served as a benchmark problem for various probabilistic techniques. 
In fact, the problem of two coloring of linear hypergraphs was one of the main motivations for introducing Local Lemma in the seminal paper of Erd\H{o}s and Lov\'{a}sz \cite{EL1975}.
It is well known that the problem of hypergraph two coloring is NP-complete \cite{Lov73} and this result holds even for hypergraphs with all edges of size 3.
In this work we present a progress on local algorithms for two coloring of uniform hypergraphs within the framework of Local Computation Algorithms.

We are going to work with uniform hypergraphs.
For the rest of the paper, $n$ is used to denote the number of vertices of considered hypergraph, $m$ its number of edges and $k$ sizes of the edges.
We assume that $k$ is fixed (but sufficiently large to avoid technical details)
    and that $n$ tends to infinity.
For a fixed hypergraph we denote by $\Delta$ its maximum edge degree.
In the instances that we are going to work with, $\Delta$ is bounded by a function of $k$ so in terms of $n$ it is $\Oh{1}$.
That implies that the number of edges $m$ is at most linear in $n$.
To avoid irrelevant technical issues we assume that considered hypergraphs do not have isolated vertices.
Then, we also have $m = \Theta(n)$.

\subsection{Local Computation Algorithms}

Rubinfeld, Tamir, Vardi and Xie proposed in \cite{fast-local} a general model of sublinear algorithms, called Local Computation Algorithms (LCA).
It is intended to capture the situation when some computation has to be made on a large instance,
but at any specific time only part of the answer is required.
Interaction with the algorithm is organized in the sequence of queries about fragments of a global solution.
The algorithm shall answer each consecutive query in sublinear time (wrt the size of the input), systematically producing a partial answer that is consistent with some global solution.
The procedures can be randomized.
As a consequence, the model allows also the algorithms to fail occasionally.

Formally, for a fixed problem, a procedure is \emph{$(t,s,\delta)$-local computation algorithm},
    if for any instance of size $n$ and any sequence of queries, it can consistently answer each of them in time $t(n)$
    using up to $s(n)$ local computation memory.
The value $\delta(n)$ shall bound the probability of failure for the whole sequence of queries.
It is usually demanded that $\delta(n)$ is small.
Local computation memory, the input, and the source of random bits are all represented as tapes with random access (the last two are not counted in $s(n)$ limit).
The local computation memory can be preserved between queries, e.g. to store some partial results determined in the previous calls.
For the precise general definition of the model we refer to \cite{fast-local}.

A procedure is called \emph{strongly local computation algorithm} if both $t$ and $s$ are at most polylogarithmic.
It is \emph{query oblivious} if the returned solution does not depend on the order of the queries.
It usually indicates, that the algorithm uses computation memory only for local computations and no information is preserved between calls.
In a followup paper \cite{ARVX}, Alon, Rubinfeld, Vardi and Xie presented generic methods of reducing the space requirements and necessary number of random bits of LCA procedures.
In the same paper, these techniques were applied to the example procedures from \cite{fast-local} converting them to strongly LCAs.
The numbers of random bits used in these algorithms have been reduced to polylogarithmic.

One of the first examples from \cite{fast-local} of a problem solved within LCA framework was finding a proper two coloring of a uniform hypergraph\footnote{
    Two coloring of a hypergraph is \emph{proper} if no edge is monochromatic.
}.
The objective was to construct a~procedure that for a given hypergraph, quickly outputs the colors of the queried vertices in such a way that,
    at any point, the returned partial coloring can be completed to a proper one.
In order to obtain satisfying results, input instances had to meet some constraints.
Not surprisingly, the constraints needed for that procedure to work in sublinear time are stronger than necessary conditions for the proper coloring to exist.

The main focus of this paper is on local computation algorithms for two coloring of uniform hypergraphs,
    that work in polylogarithmic time and with $\delta(n)= \Oh{1/n}$.
The basic version of the procedures presented here use linear amount of memory and random bits (on the average).
The techniques of reducing space and random bits from \cite{ARVX} can be also applied to our procedure.
We are going to discuss these improvements in the full version of the paper.

\subsection{Constructive LLL}

Lov\'{a}sz Local Lemma is one of the most important tools in the field of local algorithms.
Constructive aspects of the lemma are essential for our work.
For many years, Local Lemma resisted attempts to make it efficiently algorithmic. 
A breakthrough has been made by Moser in 2009.
In cooperation with Tardos, Mosers ideas has been recasted in \cite{moser-tardos} into general constructive formulation of the lemma.
They showed that,
    assuming variable setting of LLL,
    a natural randomized procedure called RESAMPLE,
    finds an instance that avoids all the bad events
    in the expected time, that in typical cases, is linear in the number of events.
This result also triggered a number of later refinements.

Adjusting RESAMPLE to LCA model remains one of the most challenging problem in the area.
At the same time, it turns out that some previous results on algorithmization of Local Lemma
    can be more easily adapted.
In fact, the algorithm for hypergraph coloring from \cite{fast-local} was based on a work of Alon \cite{alon} that in turn was built on the ideas of Beck \cite{Beck91}.

Common element of these approaches is that the assumptions of Local Lemma have to be strengthened in order for the analyzed procedures to be efficient.
Since, for our problem of two coloring, we are dealing with symmetric events,
the assumption of symmetric Local Lemma is
$
    2 p \; (\Delta+1) \; \eul < 1,
$
where, $p$ denotes the probability that, in an uniformly random coloring, a fixed edge is monochromatic in a fixed color
(i.e. $p=2^{-k}$).
The strengthened assumptions take form
\begin{equation}
\label{eq:parametrized-LLL-condition}
    2 \; p^{\alpha} \; (\Delta+1) \; \eul < 1,
\end{equation}
which for $\alpha=1$ captures the standard assumptions.
As a result, the above inequality constraints $\Delta$, and the constraint gets more restrictive as $\alpha$ gets smaller.

An important consequence of the strengthened condition \ref{eq:parametrized-LLL-condition} is that edges of the hypergraph can be trimmed to size $\alpha k$ in an arbitrary way,
    and the resulting instance would still satisfy general assumptions of Local Lemma.
In particular, it would be still two colorable.
Czumaj and Scheideler \cite{non-uniform}, utilized that flexibility to extend the work of Alon to the problem of coloring non-uniform hypergraphs.
In their work, the possibility of trimming edges has been used already at the phase of determining the components to be recolored.
A consequence of that approach is that the shapes of these components depend on the particular order in which the vertices of the hypergraph are processed.

\subsection{RESAMPLE as LCA}

Achlioptas, Gouleakis and Iliopoulos in \cite{AGI} analyzed the behavior of RESAMPLE in the context of LCA.
We refer to publicly available full version of this paper \cite{AGI-arxiv}.
They introduced additional parameter -- the number of queries that are to be answered in sublinear time.
Their analysis is most significant when this number is sublinear.
We do not follow this idea and stick to the original formulation of \cite{fast-local}.
To avoid unnecessary technical details we quote the results of \cite{AGI-arxiv} with the number of queries set to the number of vertices (that translates to setting $\alpha=1$ in Theorem 3.1 of \cite{AGI-arxiv}).
Moreover we present the statement of their main theorem for the fixed problem of two coloring of $k$-uniform hypergraphs.
That result is formulated in terms of maximum vertex degree of a hypergraph, denoted by $d$.
Parameter $\lambda$, used in the original formulation, was related to the coefficients from the assumptions of asymmetric Local Lemma.
For the problem at hand we would have $\lambda= \lambda(n) = \oh{1}$.
To avoid that dependency we formulate that theorem for any fixed (but in particular arbitrarily small) value of $\lambda$.

\begin{theorem}{(\cite{AGI-arxiv})}
    For $\eps,\lambda>0$ let $k$ and $d$ be such that $2^{1-k} (k d+1) \eul < 1- \eps$.
    Define $\zeta = \frac{\log(1/(1-\eps))}{\log{k d}}$,
    and let $\beta, \gamma >0$ be constants such that $\beta \zeta > 1+ \gamma + \lambda$.
    Then, there exists a $(n^\beta, \Oh{n},n^{-\gamma})$-local computation algorithm
        for the problem of two coloring of $k$-uniform hypergraphs
            with maximum vertex degree at most $d$.
\end{theorem}

Note that we need $\beta <1$ in order to obtain an algorithm with a sublinear time per query.
To understand better the above assumptions we discuss shortly their consequences.
Let $C = \frac{1+\gamma+\lambda}{\beta}$, note that $C>1$.
The assumptions of the theorem imply that $\zeta$ has to satisfy $\zeta > C$.
By the definition of $\zeta$, we get
$
    1/(1-\eps)>(k d)^C,
$
and hence
$
    \eps > 1-(k d)^{-C}
$.
Then, the constraint on $k$ and $d$ can be rewritten to
$
    2^{1-k} (kd+1)\; \eul < (kd)^{-C}.
$
That implies
$
    (kd)^{C+1} < \eul^{-1} 2^{k-1},
$
and hence $k d < \eul^{\frac{-1}{C+1}} 2^{\frac{k-1}{C+1}}$.
We obtain a constraint on $k d$ that is slightly more restrictive than taking $\alpha= 1/(C+1)$ in (\ref{eq:parametrized-LLL-condition}).
That shows that the instances for which the procedure answers questions in sublinear time,
    satisfy condition (\ref{eq:parametrized-LLL-condition}) with $\alpha$ strictly smaller than $1/2$.
Moreover, when $\alpha$ is close to $1/2$ we obtain quite weak bounds on the running time per query.
In particular, if we want answers in time $O(n^{1/2})$, we need $\alpha< 1/3$.
Additional property of that procedure is that the running time per query is always of the order $n^\beta$.
The previous attempts (built on \cite{alon}) answers queries in polylogarithmic time.
The same property holds for the procedure described in the current paper.
As a result, for $\alpha$ up to about $0.227$ we obtain much more efficient algorithm, then the adaptation of RESAMPLE described in \cite{AGI-arxiv}.

\section{Previous works and the main result}
\label{sec:Alon_and_main}

In order to give a fair comparison of our work with the previous results, we begin this section with sketching what
can be gained by applying some quick fixes to the hypergraphs coloring procedure from \cite{fast-local}.
It is more convenient to base our presentation on the Alons approach of \cite{alon} which was the direct inspiration for the developments of \cite{fast-local}.

All the considered procedures start with some uniformly random coloring of the vertices, called \emph{initial coloring}.
Then, an edge of the hypergraph is called \emph{$\alpha$-bad} if it (initially) contains more than $(1-\alpha)k$ vertices of one color.
We usually omit $\alpha$, when it is fixed in the context.
\emph{Bad component} is a maximal set of vertices that is spanned by a connected set of bad edges.

\subsection{Alons procedure with RESAMPLE in the second phase}

The original procedure of Alon \cite{alon} works under assumption (\ref{eq:parametrized-LLL-condition}) with $\alpha$ up to about $1/8$.
It finds a proper coloring in polynomial time (in the classical, non-local sense).
The basic variant of the procedure works in two phases.
In the first phase, it identifies some disjoint components of the hypergraph that need special attention.
In the second phase, those components are independently recolored.

The components are determined by the initial coloring in the following way.
The algorithm starts with exploring bad components.
Then, as long as there exist components that intersect the same edge, such components are merged into one.
Only the vertices of the components are recolored in the second phase.
All the remaining ones get colors given by the initial coloring.
The assumption on edge degrees, implies that the sizes of such constructed components, are at most logarithmic (whp).
At this point, most edges of the hypergraph are properly colored by the vertices outside of the components.
Every edge that is not, intersects exactly one component and for the second phase it is trimmed to that intersection.
The construction ensures that the intersections are large enough to guarantee that the hypergraph of the trimmed edges satisfies the assumptions of Local Lemma.
Therefore, a proper coloring of the components exists and can be found for every component independently.
Since the components are of logarithmic size, exhaustive search succeeds in polynomial time.

The straightforward adaptation of that basic version of the procedure to the framework of LCA yields some kind of local algorithm.
However, due to an exhaustive search applied to components, the running time of a query can be superlinear.
For that reason the authors of \cite{fast-local} used the variant sketched in \cite{alon} in which the first phase of cutting out components was iterated twice (with different parameters).
After two iterations, sizes of the subproblems that are left to solve are bounded by roughly $\log\log(n)$ so exhaustive search on them is polylogarithmic.
Unfortunately, ensuring that the components are still colorable comes with a price of stronger initial assumptions.

The described procedure can be improved by applying RESAMPLE in the second phase.
As a consequence, RESAMPLE started in bounded components, works in logarithmic time on the average.
Using standard techniques,
    we can derive a polylogarithmic bound on the running time, that ensures high probability of success across all queries.

In this variant, edge degrees of the hypergraph have to be constrained in such a way that
    Local Lemma assumptions hold after the first trimming phase
    and (whp) the components are of logarithmic size
\footnote{Relaxing the size of components to be sublinear does not make any substantial difference.}.
Sufficient conditions from \cite{alon} for these two conditions to hold take the form:
\begin{equation}
    \label{eq:AlonsConditions}
  2 \eul (\Delta+1) < 2^{\alpha k}, \\
  4 \eul \Delta^3 < 2^{(1-H(\alpha))k},
\end{equation}
where $H(.)$ is the binary entropy function.
Both these conditions constrain $\Delta$.
For small values of $\alpha$ the first inequality is stronger.
When $\alpha$ is increased, at some point the second one takes over.
The value of $\alpha$ that allows the largest $\Delta$, for $k\rightarrow \infty$ tends to the unique solution $\alpha_A$ of the equation
\[
    \alpha = \frac{1-H(\alpha)}{3}.
\]
Close to that point, both inequalities have roughly the same impact on the allowed values of $\Delta$.
For any $\alpha<\alpha_A$ and large enough $k$, 
    if the first inequality is satisfied, then second condition also holds.
This is why in the following formulation, we focus only on the first condition.
Alon in \cite{alon} commented that for $\alpha=1/8$ and sufficiently large $k$, one can allow $\Delta$ up to $2^{k/8}$.
To be more precise, we use $\alpha_A$ as the reference point for further considerations.
Its approximate value is $0.139$.
\begin{theorem}[after \cite{alon} and \cite{fast-local}]
    \label{thm:Alon}
    For every $\alpha < \alpha_A\approx 0.139$,
    and all sufficiently large $k$,
    there exists a local computation algorithm that,
    in polylogarithmic time per query,
        solves the problem of two coloring
        for $k$-uniform hypergraphs,
            in which maximum edge degree $\Delta$ satisfy
\[
  2 \eul (\Delta+1) < 2^{\alpha k}.
\]
\end{theorem}
\begin{remark}
    The above result can be further improved by
        careful trimming the bad edges to reduce the resampled area to fragments of the bad components.
    In this approach, it is possible to build witness structures that directly correspond to the structures of associated components.
    Then, the second condition of (\ref{eq:AlonsConditions}) can be replaced with
    \[
        4 \eul \Delta^2 < 2^{(1-H(\alpha))k}.
    \]
    That leads to LCA procedure that works up to $\alpha_B\approx 0.170$ in polylogarithmic time per query.
\end{remark}

\subsection{Main result}

The improved coloring procedure, presented in this paper, allows to relax the second inequality in (\ref{eq:AlonsConditions}).
As a result, our algorithm works whenever, the maximum degree satisfies
\[
  2 \eul (\Delta+1) < 2^{\alpha k}, \\
    24 \eul \Delta < 2^{(1-H(\alpha))k}.
\]
Optimizing the value of $\alpha$ we obtain the folowing theorem.

\begin{theorem}[main result]
 \label{thm:main}
    Let $\alStar$ be the unique solution of $\alpha= 1- H(\alpha)$ in $(0,1)$.
    For every $\alpha < \alStar \approx 0.227$ and all large enough $k$,
    there exists a local computation algorithm that,
    in polylogarithmic time per query,
        with probability $1-O(1/n)$ solves the problem of two coloring
        for $k$-uniform hypergraphs with maximum edge degree $\Delta$,
            that satisfies
\[
  2 \eul (\Delta+1) < 2^{\alpha k}.
\]
\end{theorem}
Within the notation of \cite{fast-local} we present $(\polylog(n), \Oh{n}, \Oh{1/n})$-local computation algorithm
    that properly colors hypergraphs satisfying the above assumption.

\subsection{Conservative \texttt{RESAMPLE}}

Since it seems hard to adapt the general procedure RESAMPLE to the LCA model, we propose a variant which,
    by incorporating the idea of edge trimming, allows for balancing the strength of the assumptions
        with the speed of expansion of the recolored fragments.
During the evaluation of RESAMPLE, a vertex is called \emph{fresh} if its color has never been resampled before.
We focus on the components of resampled vertices.
When an edge is picked to be recolored, it may contain some fresh vertices.
At that point, we can decide whether to resample the whole edge and extend the component,
    or trim the edge to the part of already resampled vertices.

\textbf{Conservative RESAMPLE:}
The procedure is parametrized by $\alpha \in (0,1]$ which controls how eager it is to extend the components of resampled vertices.
\emph{Conservative RESAMPLE} proceeds as RESAMPLE with the change that, whenever it observes a monochromatic edge $f$,
    it resamples all the vertices of $f$ only if more than $(1-\alpha)k$ of its vertices are fresh.
Otherwise it resamples only the vertices of $f$ that have been resampled before.

For $\alpha=1$, the procedure works just like the original one.
When $\alpha<1$ it is weaker, however, by standard arguments,
    one can show that it works well on instances satisfying strengthened assumptions of Local Lemma (\ref{eq:parametrized-LLL-condition}) for the same $\alpha$.
Note that, when a whole edge is resampled, the set of its fresh vertices has to be monochromatic in the initial coloring.
This explains the following observation, which in particular implies that,
    before the procedure execution starts, we can determine some bounds that will not be crossed by the algorithm.
\begin{observation}
    Conservative RESAMPLE run with parameter $\alpha$, resamples only vertices of $\alpha$-bad-edges.
\end{observation}

This observation links conservative RESAMPLE to the Alons approach.
Observe that now, the component expansions take place only inside bad components.
Moreover, disjoint components can be merged only via edges which intersects more than one component (and these edges are trimmed after the merge).
Altogether is means that, under assumptions of Theorem \ref{thm:Alon}, conservative RESAMPLE can be evaluated as LCA.

\section{Algorithm}
\label{sec:Algorithm}

We draw our inspiration from the behavior of the conservative RESAMPLE.
Our algorithm is parameterized by $\alpha \in (0,1)$ that serves similar purpose as the analogous parameter in that procedure.
We fix this value for the whole section.

Examining procedure conservative RESAMPLE we can observe that, whenever the set of resampled vertices is extended by resampling a red edge $f$,
then not only $f$ has to be initially bad, but also all its initially blue vertices have to be already contained in the set of resampled vertices.
In particular they are covered by initially bad edges.
Clearly, for every initially blue vertex $v$ of $f$, there exists an edge that caused the first resampling of $v$.
That edge extended set of resampled vertices, therefore we can assign to it a set of vertices of size $(1-\alpha)k$ that is initially monochromatic.
Altogether we obtain a structure where with the base edge $f$ we can associate a sequence of bad edges that witness the fact that $f$ became monochromatic.
We are going to make it the basic block of our construction.

\begin{definition}\label{lca-structure}
    A pair $(f, (h_1, .., h_w))$ of edge $f$ (called \emph{the base edge}) and sequence $(h_1, .., h_w)$ of its neighboring edges (called \emph{witnessing edges}) is called \emph{potential resample structure},
    iff
    \begin{enumerate}
        \item $f - \bigcup\{h_1,\ldots, h_w\}$ is initially monochromatic and has at least $(1-\alpha)k$ vertices,
        and
    \item for every $j\in [w]$, set $h_j - \bigcup\{h_1, \ldots, h_{j-1}\}$ has at least $(1-\alpha)k$ vertices in one color.
    \end{enumerate}
    In particular, for initially monochromatic edge $f$, pair $(f, \emptyseq)$ is also a potential resample structure
    (where $\emptyseq$ denotes the empty sequence).
\end{definition}

Potential resample structure will be called \emph{structure} for short.
For a structure $S$, we denote by $V(S)$ the set of vertices belonging to the edges of $S$.

Just like in the variant of Alons approach described in Section \ref{sec:Alon_and_main}, to answer a query,
    our procedure tries to build a component which is then passed to a special version of RESAMPLE.
The components are represented by subsets of vertices of the input hypergraph $H$.
This hypergraph is considered fixed for the whole computation.
For an edge $f$ we use $N(f)$ to denote the set of edges of $H$ which intersect $f$ (without $f$ itself).
An edge $f$ is a \emph{mono-tail} of a component $\rc$ if $f \setminus \rc$ is monochromatic and has at least $(1-\alpha)k$ vertices.
A structure is \emph{fresh} if the set of its vertices is disjoint from all the components constructed so far.
Structure $S$ is in distance $1$ from component $\rc$ if $V(S)$ is disjoint from $\rc$,
        but there exists an edge that intersects both $\rc$ and $V(S)$.
Note that we use the usual distance between sets of vertices of the graph
    (which is not the same as the distance between edges in the line graph of $H$ which is commonly used in other papers on the subject).
Set $V_1(\rc)$ is defined as the set of vertices of distance at most 1 from $\rc$ (i.e. the set of vertices which belong to some edge that intersects $\rc$).
Finally, we define a hypergraph $\Hrc$ as having the following sets of vertices and edges:
\[
    V(\Hrc)  = \rc, \\
    E(\Hrc)  = \{f \cap \rc: f\in E(H) \text{  s.t. } |f\cap \rc| \geq \alpha k \}.
\]

Our procedures use two data structures that are preserved between calls.
The first contains the set of vertices that has been marked as \state{colored}.
The second, called \state{current coloring}, represents the partial coloring that has been constructed so far.
When the whole computation starts, both these structures are empty.

Inspecting the initial color of a vertex is in fact an implicit reference to the random source
    in which an exclusive random bit representing the initial color is assigned to every vertex.
Similarly, testing whether given vertex belongs to a fresh structure, encapsulates references to initial colorings of some vertices and the set of marked vertices.

\SetKwFunction{FQuery}{\fun{query}}
\SetKwFunction{BuildComponent}{\fun{build\_component}}
\SetKwFunction{ColorComponent}{\fun{color\_component}}

\begin{algorithm}
  \label{alg-query}
  \DontPrintSemicolon
  \SetKwProg{Fn}{Procedure}{:}{}
  \Fn{\FQuery{$v$ - vertex}}{
        \If {$v$ is not marked as \state{colored}}  {
            \eIf{there exists a fresh structure $S$ which contains $v$}{
                $\rc \gets $ \BuildComponent{$S$}            \;    
                $\rcol \gets $ \ColorComponent{$\rc$}        \;    
                update \state{current coloring} with $\rcol$ \;
                update \state{current coloring} by assigning initial colors to 
                $V_1(\rc)\setminus \rc$  \;
                mark vertices of $V_1(\rc)$ as \state{colored}
            } {
            update \state{current coloring} by assigning to $v$ its initial color \;
            mark $v$ as \state{colored}
            }
        }
        \KwRet current color of $v$\;
  }
  \caption{Uniform Hypergraph Local Coloring - main function}{}
\end{algorithm}

\begin{algorithm}
  \label{alg-component}
  \DontPrintSemicolon
  \SetKwProg{Fn}{Procedure}{:}{}
  \Fn{\BuildComponent{$S = (f, \{h_1, .., h_w\})$ - structure}}{
        $\rc \gets \emptyset$\;
        add vertices of edges $h_1, ..., h_w, f$ to $\rc$\;
        \Repeat{component $\rc$ is \state{finished}}{
            \If{$|\rc|> \vCompBound$}{
                \FAIL \tcp*{$\rc$ is too large}
            }
            \ElseIf{there is an edge $f$ which is a mono-tail of $\rc$}{
                add vertices of $f$ to $\rc$
            }
            \ElseIf{there is a fresh structure $S'$ in distance 1 from $\rc$}{
                    add vertices of edges of $S'$ to $\rc$
            }
            \Else{$\rc$ is \state{finished}}
        }
        \KwRet $\rc$\;
    }
  \caption{Expanding component}
\end{algorithm}

\begin{algorithm}
  \label{alg-color}
  \DontPrintSemicolon
  \SetKwProg{Fn}{Procedure}{:}{}
    \Fn{\ColorComponent{$\rc$ - component}}{
        $G \gets \Hrc$ \;
        $E_t \gets |E(G)|/\Delta$ \;
        \Repeat{$2 \log(m)$ trials failed} {
            run RESAMPLE on $G$ for $2E_t$ steps \;
            \lIf{some proper coloring $\rcol$ has been found}{\KwRet $\rcol$}
        }
        \FAIL
    }
  \caption{Finding coloring inside component}
\end{algorithm}

\textbf{\fun{query}}
When a vertex is queried, either it has been inspected before and its color is already fixed, or it is fresh.
In the second case, if it does not belong to any fresh structure, we mark it as \state{colored} and return its initial color.
If it belongs to some fresh structure, we build a component on this structure, by calling \fun{build\_component}.
When the component is successfully built we try to find a proper coloring for it with procedure \fun{color\_component}.
If no failure has occurred so far, we update \state{current coloring} with the obtained component coloring
and mark all vertices of the component and its closest neighborhood as \state{colored}
(assigning initial colors for the latter ones).

\textbf{\fun{build\_component}}
New component is initialized with the vertices $V(S)$ of the given structure $S$.
Then, the component is expanded according to the rules that try to isolate a fragment of the hypergraph that can be recolored independently of other such fragments.
To achieve this, our algorithm expands the component in two ways.
First, it adds mono-tails as long as it can find some.
Then it searches for a fresh structure that is close to (but disjoint with) the component constructed so far,
    and add all the vertices of that structure to the component.
These two steps are repeated as long as the component grows.
However, if at some point, the size of the component exceeds the prescribed bound, the procedure gives up and declares failure.
In such a case the whole computation fails as well.
The important property is that, under assumptions of Theorem \ref{thm:main},  large components are unlikely to be built.
If the component is completed without exceeding this bound, it is returned for further processing.

\textbf{\fun{color\_component}}
When component $\rc$ is successfully determined, we attempt to color the trimmed hypergraph $\Hrc$ built on this component.
We use procedure RESAMPLE that is run for the number of steps (resamples of edges) determined by the number of edges (possible bad events) in $\Hrc$.
If a proper coloring is found, it is returned from the procedure.
If not, we run RESAMPLE again (starting from a new random coloring),
    allowing $2\log(m)$ of such trials.
When none of them is successful,
the procedure declares failure, which again invalidates the whole computation.
We will argue that this event is unlikely to happen as well.
We rely here on a very weak (Markov type) bounds for the probability that RESAMPLE does not succeed in a long run.
It is clear that using more precise estimates, would allow to derive procedures with better running times.
Since that change would not make any substantial difference for the overall result,
    we stick to the most simple techniques.

\section{Proof}

We show that our algorithm fulfills the conditions of LCA for the problem of two coloring of uniform hypergraphs.
A sequence of calls to \fun{query} caused by a sequence of queries is called \emph{a computation}.
Computation is \emph{successful} if no failures occur during its evaluation.
It is \emph{complete} if additionally every vertex of the hypergraph is queried.
First, we prove that, for every successful computation, the returned coloring can be extended to a proper one
and
that the colors returned for multiple queries of the same vertex are consistent with previous answers.
Then, we argue that failures occur rarely.
Finally, we estimate running time of a single call to \fun{query} and the amount of memory needed to carry out the whole computation.

\subsection{Correctness}
\label{ss:correctness}

When answering a query, our algorithm sometimes assigns colors not only to the queried vertex but also to a larger set of vertices (those eventually marked as \state{colored}).
It is an important property of the procedure, that such colors are not changed when responding to subsequent queries.
As a consequence the answers for vertices queried multiple times are consistent.
Observe that new colors are assigned to vertices only in \fun{color\_component},
    thus only vertices added to components can be recolored.
Hence, it remains to show that once a vertex is marked \state{colored}, it cannot be added to a component constructed during subsequent queries.

Vertex becomes \state{colored} either after being included in a new component,
    or when it becomes a close neighbor of just constructed component,
    or when (at the time of being queried) it lies outside of any fresh structure.
According to Proposition \ref{prop:comp-disjoint} proved in the next section, the components constructed during the computation are of distance at least $2$.
In particular they are disjoint.
Moreover for any component $\rc$ no vertex $v$ from the neighborhood of $\rc$ can belong to some other component.
That implies that, unless $v\in \rc$ it is not going to be recolored.
By Proposition \ref{prop:structureCover} (also proved in the next section) every component can be covered by structures.
Since the components are disjoint, these structures have to be fresh at the time when the component is built.
Therefore, a vertex which, at the time of being queried, does not belong to any fresh structure,
    cannot be included in a component constructed in subsequent queries.

\subsubsection{Independence of components}

The procedure of component building has been deliberately defined in a way to satisfy the following property.
\begin{proposition}
    \label{prop:structureCover}
    In each step of evaluation of \fun{build\_component}, the component built so far can be covered by structures.
\end{proposition}

\begin{proof}
    The property is satisfied at the beginning, since the procedure is started from a structure.
    Expanding via fresh structures is an obvious case.
    The remaining one is when the component is extended with a  mono-tail.
    When adding mono-tail $f$ to component $\rc$,
    let $E(\rc)$ be the set of edges used to extend component $\rc$ so far.
    We show that $(f, N(f) \cap E(\rc))$ can be arranged to form a structure if we appropriately order the witnessing edges.
    As $f$ was added to $\rc$ as mono-tail, we have that $f \setminus \rc$ is monochromatic and large enough.
    Every edge of $g\in N(f) \cap E(\rc)$ added at least $(1-\alpha)k$ new vertices to the constructed component.
    By design, these sets of vertices were either monochromatic (if $g$ was added as mono-tail) or belonged to a fresh structure.
    In both of these cases the set of the new vertices added to the component by $g$ contains at least $(1-\alpha)k$ vertices of one color.
    By definition these sets are disjoint for distinct edges of $N(f) \cap E(\rc)$.
    That shows that the edges from $N(f) \cap E(\rc)$, ordered according to the times they were used to expand $\rc$, satisfy the requirements of the sequence of witnessing edges.
\end{proof}

We can now prove the following proposition which in particular implies that the components can be recolored independently,

\begin{proposition} \label{prop:comp-disjoint}
    Any two distinct components constructed during a computation are in the distance at least 2.
\end{proposition}

\begin{proof}
    Every new component is started from a fresh structure.
    That implies that at that moment it is disjoint from all the components constructed so far.
    Suppose that at some point (possibly just after adding the first fresh structure),
        currently constructed component $\rc$ becomes of distance 1 from some previously built component $\rc'$.
    Let $v\in \rc$ be a vertex of distance 1 from $\rc'$ (it means that there exists an edge that contains $v$ and intersects $\rc'$).
    By the previous proposition there exists a structure $S_v$ for which set $V(S_v)$ contains $v$ and is a subset of $\rc$.
    Clearly, $S_v$ is of distance 1 from $\rc'$, and at a time when the construction of $\rc'$ was finished, $S_v$ was fresh.
    That contradicts the extension rules of procedure \fun{build\_component}, which would use $S_v$ to extend $\rc'$.
\end{proof}

\subsubsection{Proper coloring}
\label{ss:properColoring}

We want to show that, for any successful computation, the obtained coloring can be completed to a proper one.
By the independence of coloring of the components,
    we can focus on the formulation that any complete computation returns a proper coloring of the hypergraph.
From Proposition \ref{prop:comp-disjoint} we know that each edge intersects at most one component.
Hence, we can partition the edges of $H$ into three families:
\begin{itemize}
    \item $E_1$ -- edges disjoint with components.
    \item $E_2$ -- edges with small intersection with some component (i.e. of size less than $\alpha k$),
    \item $E_3$ -- edges with large intersection with some component (i.e. of size at least $\alpha k$),
\end{itemize}
For each of the families we explain why these edges cannot be monochromatic in the final coloring produced by the computation.

The edges of the first family are entirely colored by the initial coloring.
Hence, if some edge of $E_1$ is monochromatic in the final coloring, then it is also initially monochromatic.
Such a monochromatic edge itself forms a fresh structure, so it would be either added to some neighboring component or a new component would be built on it.
Both this cases contradict the definition of $E_1$.

Every edge of $E_2$, has more than $(1-\alpha)k$ vertices outside of the component it intersects.
As we already observed, for any such edge $f$, the set of its vertices $f'$  that do not belong to any of the constructed components are colored by the initial coloring.
Since the component that intersects $f$ has not been extended by $f$, 
    this edge is not mono-tail, thus set $f'$ cannot be monochromatic in the initial coloring.
Therefore, $f$ is not monochromatic in the final coloring as well.

A proper coloring for the edges of $E_3$ is ensured by procedure \fun{color\_component},
since these edges were considered (possibly in a trimmed form) in some $\Hrc$.

\subsection{Probability of failure}
\label{ss:failure}

For each query there are two places in the algorithm in which it can fail:
\begin{enumerate}
    \item inside \fun{build\_component} when the constructed component becomes too large,
    \item during \fun{color\_component} when each of several tries of RESAMPLE ends without finding a proper coloring.
\end{enumerate}
We bound the probability of failing in the following propositions.

\begin{proposition} \label{main-prop}
    Under the assumptions of Theorem \ref{thm:main}, for any sequence of queries,
        the probability that the algorithm
        starting from uniformly random initial coloring,
            builds a  component of size at least $\vCompBound$ is smaller than $\frac{2}{m}$.
\end{proposition}

The proof of this proposition is moved to Appendix \ref{sec:compSize}.

\begin{proposition} \label{coloring-prop}
    Under the assumptions of Theorem \ref{thm:main}, for any sequence of queries,
        the probability that the algorithm
        starting from uniformly random initial coloring,
            fails in a call to \fun{color\_component} is less than $\frac{1}{m}$.
\end{proposition}
\begin{proof}
Consider procedure \fun{color\_component} working on a component $\rc$.
Hypergraph $\Hrc$ built on $\rc$ have
    edges of size at least $\alpha k$.
Hence, the probability that an edge is monochromatic in the random coloring (which correspond to the occurrence of a bad event) is at most $2^{1 - \alpha k}$.
The upper bound on edge degrees is inherited from the input hypergraph, i.e. the maximum edge degree of $\Hrc$ is at most $\Delta$.
By the assumptions of Theorem \ref{thm:main} we have $\eul\; 2^{1 - \alpha k} (\Delta+1) < 1$
    and hence the assumptions of the symmetric Local Lemma are satisfied.
By the results of \cite{moser-tardos} the expected running time of RESAMPLE (in the number of steps) is $E_t = |E(\Hrc)|/\Delta$.
Therefore, if we run that procedure for $2E_t$ steps, the probability of not finding a proper coloring is at most $1/2$.

Repeating RESAMPLE (with the bounded number of steps) up to $2\log(m)$ times,
    each time starting with a new random initial coloring,
    we obtain that the probability of failure is not greater than $2^{-2\log(m)} = m^{-2}$.
During the whole algorithm there are at most $m$ calls to \fun{color\_component}.
Hence, the probability that we observe a failure in one of these calls is at most $1/m$.
\end{proof}

\subsection{Running time and space}

By explicit checks of failure conditions, the running time of a query is always polylogarithmic.
As we treat $k$ and $\Delta$ as constants, inspecting any neighborhood (of a constant distance) of a vertex or edge is made in time $\Oh{1}$.
In particular, for any set of vertices $A$,
    determining whether there exists a structure of any constant distance from $A$ can be done in time $\Oh{|A|}$.
Therefore looking for a structure containing a queried vertex $v$ in a call to \fun{query} is done in a constant time.
Moreover, procedure \fun{build\_component} performs at most $\Oh{\log(m)}$ iterations and in each iteration
    looking for objects satisfying the expansion rules takes time $\Oh{\log(m)}$.
Note that the components passed to \fun{color\_component} have at most logarithmic sizes,
    therefore one run of RESAMPLE require $\Oh{\log(m)}$ steps.
Each step can be processed in $\Oh{1}$ time.
Then, the whole procedure (with $\Oh{\log(m)}$ trials) takes time $\Oh{\log^2(m)}$.
In total we need $\Oh{\log^2(m)}$ time per query.

Regarding space -- we chose to remember states and colors of \state{colored} vertices.
Thus we need $\Oh{n}$ space that is preserved between calls.
Additionally, for each query we need $\Oh{\log(m)}$ temporary memory.

\subsection{Completing the proof}

\begin{proof}[Proof of Theorem \ref{thm:main}]
    We argued in Section \ref{ss:correctness} that a successful evaluation of the algorithm described in Section \ref{sec:Algorithm} constructs a coloring that can be completed to a proper one.
    Moreover, the answers for vertices queried multiple times are consistent and each query is answered in polylogarithmic time.
    It remains to show that the computation is indeed successful with large probability.
    As explained in Section \ref{ss:failure} there are two possible reasons of failure.
    It happens when a component becomes too large in \fun{build\_component}
        or when multiple tries of running RESAMPLE do not find a proper coloring in \fun{color\_component}.
    Propositions \ref{main-prop} and \ref{coloring-prop},
        and the fact that we consider hypergraphs with $n$ and $m$ of the same order,
        imply that with probability $1-O(1/n)$ none of these events take place.
    Hence the computation fails with probability $O(1/n)$.
\end{proof}

\section{Final thoughts}

It seems natural to expect that the described result can be improved by considering deeper potential resample structures.
Indeed, the witnessing edges that are not initially monochromatic also need some other edges to be recolored
        before they become a candidates for resampling.
That intuition is misleading.
Our analysis shows that the most constraining factor is the probability that an edge contains at least $(1-\alpha)k$ vertices of one color
    (which is the condition demanded from the witnessing edges).
As a result the most pessimistic case is when the resampled components consists mainly of witnessing edges
    (i.e. the number of base edges is as small as possible).
Deeper structures does not essentially influence this pessimistic proportion
    (just like the proportion between numbers of leaves and of inner nodes in regular trees is almost the same for all sizes).

Finally, we remark that it is possible to make the presented algorithm query oblivious and strongly LCA.
Additionally, using techniques of \cite{ARVX}, the number of necessary random bits can be reduced to polylogarithmic value.
We are going to describe these improvements in detail in the full version of this paper.

\bibliographystyle{amsplain}    
\bibliography{pb}

\providecommand{\bysame}{\leavevmode\hbox to3em{\hrulefill}\thinspace}
\providecommand{\MR}{\relax\ifhmode\unskip\space\fi MR }
\providecommand{\MRhref}[2]{%
  \href{http://www.ams.org/mathscinet-getitem?mr=#1}{#2}
}
\providecommand{\href}[2]{#2}
\begin{thebibliography}{10}

\bibitem{AGI}
Dimitris Achlioptas, Themis Gouleakis, and Fotis Iliopoulos, \emph{Simple local
  computation algorithms for the general lov\'{a}sz local lemma}, Proceedings
  of the 32nd ACM Symposium on Parallelism in Algorithms and Architectures (New
  York, NY, USA), SPAA '20, Association for Computing Machinery, 2020,
  p.~1–10.

\bibitem{AGI-arxiv}
Dimitris Achlioptas, Themis Gouleakis, and Fotis Iliopoulos, \emph{Simple local
  computation algorithms for the general lovasz local lemma}, 2020.

\bibitem{alon}
Noga Alon, \emph{A parallel algorithmic version of the local lemma}, Random
  Structures Algorithms \textbf{2} (1991), no.~4, 367--378. \MR{1125955}

\bibitem{ARVX}
Noga Alon, Ronitt Rubinfeld, Shai Vardi, and Ning Xie, \emph{Space-efficient
  local computation algorithms}, Proceedings of the {T}wenty-{T}hird {A}nnual
  {ACM}-{SIAM} {S}ymposium on {D}iscrete {A}lgorithms, ACM, New York, 2012,
  pp.~1132--1139. \MR{3205279}

\bibitem{Beck91}
J\'{o}zsef Beck, \emph{An algorithmic approach to the {L}ov\'{a}sz local lemma.
  {I}}, Random Structures Algorithms \textbf{2} (1991), no.~4, 343--365.
  \MR{1125954}

\bibitem{non-uniform}
Artur Czumaj and Christian Scheideler, \emph{Coloring non-uniform hypergraphs:
  a new algorithmic approach to the general {L}ov\'asz local lemma},
  Proceedings of the {E}leventh {A}nnual {ACM}-{SIAM} {S}ymposium on {D}iscrete
  {A}lgorithms ({S}an {F}rancisco, {CA}, 2000), ACM, New York, 2000,
  pp.~30--39. \MR{1754838}

\bibitem{EL1975}
Paul Erd{\H{o}}s and L{\'a}szl{\'o} Lov{\'a}sz, \emph{Problems and results on
  {$3$}-chromatic hypergraphs and some related questions}, Infinite and finite
  sets ({C}olloq., {K}eszthely, 1973; dedicated to {P}. {E}rd{\H o}s on his
  60th birthday), {V}ol. {II}, Colloquia Mathematica Societatis J{\'{a}}nos
  Bolyai, vol.~10, North-Holland, Amsterdam, 1975, pp.~609--627.

\bibitem{Lov73}
L{\'a}szl{\'o} Lov{\'a}sz, \emph{Coverings and coloring of hypergraphs},
  Proceedings of the {F}ourth {S}outheastern {C}onference on {C}ombinatorics,
  {G}raph {T}heory, and {C}omputing ({F}lorida {A}tlantic {U}niv., {B}oca
  {R}aton, {F}la., 1973), 1973, pp.~3--12.

\bibitem{moser-tardos}
Robin~A. Moser and G{\'a}bor Tardos, \emph{A constructive proof of the general
  {L}ov\'asz {L}ocal {L}emma}, J. ACM \textbf{57} (2010), no.~2, Art. 11, 15.
  \MR{2606086}

\bibitem{fast-local}
Ronitt Rubinfeld, Gil Tamir, Shai Vardi, and Ning Xie, \emph{N.: Fast local
  computation algorithms}, ICS 2011, 2011, pp.~223--238.

\end{thebibliography}

\appendix

\section{Bounding component sizes}
\label{sec:compSize}
It is a common aaproach to define some coupling that associate witness structures met during evaluation of the algorithm with random trees generated by a Galton-Watson process.
In that language, the second of Alons assumptions (\ref{eq:AlonsConditions}) can be interpreted as a condition that the associated process is subcritical.
Our analysis do not follow that path, since it relies on the amortization of two kinds of basic events that are entailed by the event described by a witness structure.

For the sake of analysing procedure \fun{build\_component}, we alter slightly its behaviour by allowing it to build arbitrarily large components.
Then, the event that the unrestricted variant constructs a large component corresponds to the event that the original variant fails.
For the rest of this section \fun{build\_component} denotes the modified version.

In order to show that constructing a large component is unlikely, we associate each component build in a call to \fun{build\_component} with an event determined by the initial coloring.
Following common technique, these events are described by structures called witness trees.
We show that for all possible witness trees of a certain size, the probability that any of the corresponding events holds is small.

\subsection{Witness tree - basic properties} \label{sec-witness}

A witness tree can be identified with a subgraph of the line graph\footnote{
    The line graph $L(H)$ is the graph built on $E(H)$ in which two distinct vertices (representing edges of $H$) are adjacent if the corresponding edges intersect.
} $L(H)$ of $H$, equipped with some extra information.

\begin{definition}
    \emph{Witness tree} is a triple $(T,\lab, \ort)$ in which
\begin{enumerate}
    \item $T$ is a tree which is a subgraph of the line graph $L(H)$,
    \item $\lab$ assigns labels from $\{\type{W}, \type{M}, \type{J}\}$ to the vertices of $T$,
    \item $\ort$ is a function that directs the edges of $T$.
\end{enumerate}
\end{definition}

The base structure of a witness tree, a tree $T$, is a subgraph of $L(H)$,
    so its vertices are actually edges of $H$.
That can be quite misleading, therefore
    we prefer to work as if $T$ was just a separate tree paired with an isomorphism into a subgraph of $L(H)$.
We use term \emph{nodes} for the vertices of witness trees, to keep a clear distinction between them and the vertices of $H$.
There is no need to explicitly represent the isomorphism.
Instead, we say about the edge represented by or corresponding to a node.

Size of a witness tree is the number of its nodes.
It is known fact that an infinite $\Delta$-regular rooted tree contains
$
    \frac{1}{(\Delta-1) u +1} \binom{\Delta u}{u} < (\eul \Delta)^u
$
rooted subtrees of size $u$.
Therefore, the number of trees of size $u$, that are subgraphs of $L(H)$ and contain a node that corresponds to a specific edge of $H$, is smaller than $(\eul \Delta)^u$.
This implies that the number of all subtrees of $L(H)$ of size $u$ is smaller than $m(\eul \Delta)^u$.
For a given tree $T$ with $u$ nodes,
    there are $3^u$ candidates for labeling functions $\lab$,
    and $2^{u-1}$ choices for $\ort$.
That allows to bound the total number of witness trees of size $u$ in the following way.
\begin{corollary}
    \label{cor:counting-witness-trees}
    The number of witness trees of size $u$ is smaller than $m(6 \eul \Delta)^u$.
\end{corollary}

Edge orientation within a witness tree can be used to derive a partial order on its nodes.
Nodes without outgoing edges are called \emph{roots}.
Then, for each node $x$ we define its \emph{depth} as the length of the maximum directed path from $x$ to one of the reachable roots.
Roots have depth $0$.

Let $\tau$ be a witness tree.
For a positive $d$, let $V_{<d}(\tau)$ be the set of vertices belonging to the edges assigned to the nodes of $\tau$ of depth strictly smaller than $d$.
With every node $x$ of $\tau$ with type \type{M} or \type{W},
we associate set of vertices $w(x)$ defined as follows
\[
    w(x) = f_x \setminus V_{<d}(\tau),
\]
where $d$ is the depth of $x$, and $f_x$ is the edge represented by $x$.
Note that, for any pair of nodes $x$ and $y$ of different depths, sets $w(x)$ and $w(y)$ are disjoint.

\subsection{Events described by witness trees}

For a witness tree $\tau$, event $\event{\tau}$ is defined as the conjunction of the following basic events associated to nodes of $\tau$:
\begin{itemize}
    \item for every node $x$ of type \type{M}, set $w(x)$ is initially monochromatic,
    \item for every node $y$ of type \type{W}, set $w(y)$ has initially at least $(1-\alpha)k$ vertices of one color.
\end{itemize}
Observe, that $\event{\tau}$ is determined by the initial coloring.
Note also that the definition of $\event{\tau}$ completely ignores nodes of type \type{J}.

\subsection{Proper witness trees}

Witness trees are used to describe events determined by the initial coloring.
The nature of the events of our interest imposes additional properties.
Most important of these are captured by the following definition.

\begin{definition}
    A witness tree is \emph{proper} if all of the following conditions hold:
    \begin{enumerate}[(i)]
        \item it has more nodes of type \type{M} than nodes of type \type{J},
        \label{cond-more-M-than-J}
        \item nodes of type \type{M} or \type{W} with the same depth, correspond to pairwise disjoint edges,
        \label{cond-same-depth}
    \item every node $x$ of type \type{M} or \type{W}, satisfies $|w(x)| > (1-\alpha)k$.
        \label{cond-big-tail}
    \end{enumerate}
\end{definition}

The nodes of type \type{J} are very weakly constrained by the above definition.
It reflects the fact that their sole purpose is to keep a whole structure connected.

Condition (\ref{cond-same-depth}) in the definition of a proper witness structure,
    augments the definition of $w(.)$ so that, for a proper structure, all these sets are disjoint.
That justifies the following observation.

\begin{observation} \label{obs-proper-independent-events}
    For a proper witness tree $\tau$, all the basic events associated with the nodes of $\tau$ are mutually independent.
\end{observation}

\subsection{Constructing a structure for a component} \label{sec-construct-structure}

We show how to associate a proper witness tree of related size to
    every component constructed during a computation.
It is discovered simultaneously with the progress of \fun{build\_component}.
A crucial property is that
    for every component we assign a witness tree $\tau$,
    for which $\event{\tau}$ holds in the initial coloring.

Consider an evaluation of \fun{build\_component}. 
We focus on edges which causes growth of the component - mono-tails and edges from fresh structures.
We sort them in the order in which our procedure adds their vertices to component,
assuming that for a fresh structure we first consider witnessing edges (according to their indices),
    and then the base edge (which at that point can be classified as a mono-tail).
Then, we can associate with each run of the \fun{build\_component} a sequence of extending edges $(f_1, \ldots, f_r)$.

Our construction of witness tree follows such sequence $(f_1, \ldots, f_r)$,
    adding nodes related to $f_i$ one by one.
During this process, we do not require that the current structure is connected,
    so we refer to it as \emph{forest}.
We use term \emph{subcomponent} for
    connected components of the forest constructed so far.
We say that an edge $f$ interects a subcomponent if the subcomponent contains a node, which represents an edge that intersects $f$.
Note that, when a new node is added to the forest, all its outgoing edges are declared 
    at this moment and remain unchanged till the end of construction.
This implies that the depth of the inserted node, cannot be changed in the subsequent steps.

Starting from the empty forest,
for every consecutive edge $f_j$ we proceed as follows:
\begin{enumerate}
    \item add to the forest a new node $x_j$ representing $f_j$;
    \item label $x_j$ as
    \begin{itemize}
        \item \type{W}, if $f_j$ extended component as a witnessing edge,
        \item \type{M}, if $f_j$ extended component as a mono-tail edge;
    \end{itemize}
\item add edges between $x_j$ and nodes added to the forest so far, in the following way:
    \begin{enumerate}
        \item \label{structure-add-disjoint} if $f_j$ is disjoint with all previously considered edges, 
                then leave $x_j$ without edges (it becomes the root of the new subcomponent),
            \item \label{structure-add-intersecting} otherwise,
                for each subcomponent containing some nodes representing neighbors of $f_j$,
                select such a node with the greatest depth (in case of a tie, choose any one),
                and add to the forest an edge from $x_j$ to that node.
    \end{enumerate}
\end{enumerate}

Note that, case (\ref{structure-add-disjoint}) can happen only when processing edges that form a fresh structure.
Observe also, that in case (\ref{structure-add-intersecting}), if at least two subcomponents are involved, then these subcomponents are merged into one.

When the above construction is finished the forest may still not be connected.
However, since each new subcomponent starts from a fresh structure which is in distance 1 from the component constructed so far,
    for each subcomponent there exists an edge which intersects it and one of the previously created subcomponents.
Hence, we repeat the following steps, as long as there are more than one subcomponent:
\begin{enumerate}
    \item pick any edge $f$ which intersects two (or more) subcomponents,
    \item create a node of type \type{J}, which corresponds to $f$,
    \item add outgoing edges to each subcomponent it intersects.
\end{enumerate}
In the last step it does not matter to which nodes inside subcomponent we attach edges, but for consistency we choose nodes of the greatest depths.
In this way a construction of the witness structure is completed.

\begin{observation} \label{obs-proper-structure}
    The witness tree constructed for the component is proper.
\end{observation}

\begin{proof}
    Condition (\ref{cond-more-M-than-J}) is satisfied,
        since each node of type \type{J} merges together at least two subcomponents,
        and each of them contains at least one node of type \type{M}
        (the node that corresponds to the base edge of the structure that started the subcomponent).
    The second condition (\ref{cond-same-depth}) follows from the fact that when adding a node $x_j$, edge attaching rules ensure that the depth of $x_j$ is strictly greater than the depths of the nodes
    corresponding to the edges from $\{f_1, \ldots, f_{j-1}\}$ that intersect $f_j$.
    The last condition (\ref{cond-big-tail}) is a consequence of the following claim.
\begin{claim} \label{cl-preserve-dependency}
    Let $\rc_j = \bigcup \{f_1, \ldots, f_{j-1}\}$, i.e. the set of vertices of the component just before it was extended with $f_j$.
    For every node $x_j$ of type \type{M} or \type{W} added to the constructed witness tree $\tau$ we have
$
    w(x_j) = f_j \setminus \rc_j.
$
    Moreover, the size of $w(x_j)$ is strictly greater than $(1-\alpha)k$.
\end{claim}
    Let $D_j = N(f_j) \cap \{f_1, \ldots, f_{j-1} \}$ and $d$ be the depth of $x_j$ in $\tau$.
    It is obvious that $f_j \cap \rc_j = f_j \cap V(D_j)$.
    Observe that nodes corresponding to the edges of $D_j$ have depth strictly smaller than $d$ so $V(D_j) \subset V_{<d}(\tau)$.
    Since for any edge $f_l$ which follows (i.e. $l>j$) and intersects $f_j$,
        the depth of $x_l$ is strictly greater than $d$, it holds that $f_j \cap V(D_j) = f_j \cap V_{<d}(\tau)$.
    This implies that $w(x_j) = f_j \setminus V_{<d}(\tau) = f_j \setminus V(D_j) = f_j \setminus \rc_j$.
    To complete the proof, notice that $f_j \setminus \rc_j$ is a set of vertices by which $f_j$ extends a component and,
        by design of \fun{build\_component}, it has size more than $(1-\alpha)k$ for both witnessing and mono-tail edges.
\end{proof}

We can now extend above observation and show the desired property.

\begin{proposition} \label{prop:proper-event-holds}
    Let $\tau$ be a witness tree constructed for a component build by \fun{build\_component},
        then $\tau$ is proper and $\event{\tau}$ holds in the initial coloring.
\end{proposition}

\begin{proof}
    The proof is a simple extension of the previously stated argument.
    From Claim \ref{cl-preserve-dependency} we know that for a node $x$ of type \type{M} or \type{W},
        set $w(x)$ is exactly the set of vertices by which the corresponding edge extended the component.
    Then, if we look at definitions of witnessing and mono-tail edges,
        we see that they are equivalent to occurrence of appropriate event on set $w(x)$.
\end{proof}

\subsection{Probability of the event described by a proper witness tree}
We start with bounding the probabilities of the basic events associated with nodes of a proper witness structure.

\begin{proposition}
    For a set of vertices $S$ satisfying $(1-\alpha)k \leq |S| \leq k$ we have
    \begin{align}
        \prt{$S$ is initially monochromatic} &\leq 2^{1-(1-\alpha)k},
        \label{eq-bound-mono-tail}
        \\
        \prt{$S$ has initially at least $(1-\alpha)k$ vertices of the same color} &\leq 2^{1-(1-H(\alpha))k},
        \label{eq-bound-bad-tail}
    \end{align}
    where $H(x)$ is the binary entropy function.
\end{proposition}

\begin{proof}
    The first bound is straightforward -- the probability that all vertices in $S$ have the same initial color is $2/2^{|S|}$.
    For the second one, observe that having at least $(1-\alpha)k$ vertices in one color is equivalent to having at least $g = \ceil{(1-\alpha)k}$ of such vertices.
    Let $s = |S|$.
    Now the left hand side of (\ref{eq-bound-bad-tail}) is at most 
    \[
      \frac{2}{2^s}\sum_{i=g}^{s} \binom{s}{i} =
      \frac{2}{2^s}\sum_{i=0}^{s-g} \binom{s}{i}.
    \]
    The above expression is increasing with $s$, so to get an upper bound we assign to $s$ the maximum possible value $k$.
    The bound becomes
    \[
      \frac{2}{2^k} \sum_{i=0}^{k-g} \binom{k}{i}  =
      \frac{2}{2^k} \sum_{i\leq \alpha k} \binom{k}{i} \leq
      \frac{2}{2^k} 2^{H(\alpha)k},
    \]
    which completes the proof.
\end{proof}

Applying this proposition to sets $w(x)$ defined by a witness tree,
    and using lower bound on size of this sets in proper witness trees
        (condition (\ref{cond-big-tail}) of the corresponding definition),
    we obtain bounds for the probabilities of the basic events.

\begin{corollary}
    For a node of a proper witness tree, the probability of the basic event corresponding to that node is
    \begin{itemize}
        \item at most $\pMonoTail = 2^{1 - (1-\alpha)k}$, if the node is labeled with \type{M},
        \item at most $\pBadTail = 2^{1 - (1-H(\alpha))k}$, if the node is labelled with \type{W}.
    \end{itemize}
\end{corollary}

Let $\tau$ be a proper witness tree of size $u$ and let $N_M(\tau), N_W(\tau), N_J(\tau)$ denote the numbers of nodes respectively of types \type{M}, \type{W}, and \type{J} in $\tau$.
Using the fact that the basic events defined by a proper witness tree are independent (Observation \ref{obs-proper-independent-events}), we get
\begin{equation} \label{eq-event-simple-bound}
    \pr{\event{\tau}} \leq \pBadTail{}^{N_W(\tau)} \pMonoTail{}^{N_M(\tau)}.
\end{equation}
Condition (\ref{cond-more-M-than-J}) of the definition of a proper witness tree
        implies that $N_M(\tau) > N_J(\tau)$. 
Therefore
    \[
        \pr{\event{\tau}} \leq \pBadTail{}^{u-2N_M(\tau)} (\pMonoTail{}^{1/2})^{2N_M(\tau)}.
    \]

\begin{corollary}
\label{cor:event-probability}
    Let $q = \max(\pBadTail,\pMonoTail{}^{1/2})$.
    For a proper witness tree $\tau$ of size $u$, we have
    \[
        \pr{\event{\tau}} \leq q^u.
    \]
\end{corollary}

\subsection{Proof of Proposition \ref{main-prop}}

We start with a technical statement that sums up the results of this section, obtained so far.

\begin{proposition}
    \label{prop:Deltaq}
    Let $q = \max(\pBadTail,\pMonoTail{}^{1/2})$.
    Assume that
    \[
        6 \eul \Delta q < \frac{1}{2}.
    \]
    Then, for any sequence of queries,
        the probability that the algorithm, starting from uniformly random initial coloring,
            builds a component of size greater than $\vCompBound$,
        is less than $\frac{2}{m}$.
\end{proposition}

\begin{proof}
    Observe that,
    if
        our algorithm builds a component of size greater than $\vCompBound$,
    then, by Proposition \ref{prop:proper-event-holds},
        there exists a witness tree $\tau$ of size at least $u = 2\log(m)$,
        which
            is proper and for which $\event{\tau}$ holds in the initial coloring.
    Therefore, the probability of building a large component is bounded by
        the probability that there is a proper witness tree of size at least $u$ for which the corresponding event holds.
    We bound the latter by estimating the expected number of such trees.
    For $q = \max(\pBadTail,\pMonoTail{}^{1/2})$, by linearity of expectation we get
    \[
        \Et{\# proper witness trees $\tau$ of size at least $u$ for which $\event{\tau}$ holds}
        \leq \sum_{j \geq u} m(6 \eul \Delta)^j \cdot q^j,
    \]
    (we used Corollaries \ref{cor:counting-witness-trees} and \ref{cor:event-probability}).
    Since from assumptions
    $
        6 \eul \Delta q < \frac{1}{2}
    $
    and $u = 2\log(m)$,
    we got that this probability is bounded by
    $
        2m (1/2)^u = 2/m
    $.
\end{proof}

\begin{proof}[Proof of Proposition \ref{main-prop}]
    For the sake of clarity, for the rest of this proof, we make dependence of considered parameters on $\alpha$ explicit.
    By Proposition \ref{prop:Deltaq}, condition
    \begin{equation}
    \label{eq:Dq}
        6 \eul \Delta q(\alpha) < 1/2.
    \end{equation}
    is sufficient to derive that the probability of building a large component is less then $2/m$.
    Recall that $q(\alpha) = \max(\pBadTail(\alpha),\pMonoTail(\alpha)^{1/2})$ and
    \[
        \pMonoTail(\alpha)^{1/2} = 2^{(1 - (1-\alpha) k)/2}\;\;\; \text{  and }\;\;\; \pBadTail(\alpha)= 2^{1 - (1-H(\alpha))k}.
    \]
    Therefore, if for some specific $\alpha'$ it holds $\frac{1-\alpha'}{2}> (1-H(\alpha'))$,
        then also for large enough $k$ we have $\pMonoTail(\alpha')^{1/2} < \pBadTail(\alpha')$.
    The unique solution of $\frac{1-\alpha}{2}= (1-H(\alpha))$ in interval $(0,1/2)$ is $\alpha_0 \approx 0.133$.
    Then, for every $\alpha \in (\alpha_0, 1/2)$ and all sufficiently large $k$, we have $q(\alpha)=\pBadTail(\alpha)$.
    For the values of $\alpha$ in that interval, condition (\ref{eq:Dq}) becomes
    \begin{equation}
        \label{eq:DeltaHal}
        6 \eul \Delta 2^{1 - (1-H(\alpha))k} < 1/2.
    \end{equation}
    Note that Theorem \ref{thm:main} explicitly assumes that, $\alpha < \alStar$ and
    \begin{equation}
        \label{eq:assumedDelta}
        \Delta +1 < \frac{ 2^{\alpha k}}{2 \eul}.
    \end{equation}
    Observe that if $\alpha < 1- H(\alpha)$, then for large enough $k$, inequality (\ref{eq:assumedDelta}) implies (\ref{eq:DeltaHal}).
    Condition $\alpha < 1- H(\alpha)$ holds for small $\alpha$,
    and $\alStar$ is precisely the point where the inequality changes.
    Therefore for any $\alpha$ that satisfies $\alpha_0 < \alpha < \alStar$, and all large enough $k$,
    condition (\ref{eq:Dq}) is implied by (\ref{eq:assumedDelta}).
    
    To complete the proof, note that, when $\alpha \leq \alpha_0$ then assumptions of Theorem \ref{thm:main} are also satisfied for any $\alpha' \in (\alpha_0, \alStar)$ and in this case we can run our procedure with $\alpha'$ instead of $\alpha$.

\end{proof}

\end{document}